\documentclass[journal,11pt,onecolumn, draft]{IEEEtran}
\usepackage{amsfonts, amsmath, amsthm, amssymb}
\usepackage{graphicx}
\usepackage{listings}
\usepackage[margin=1in]{geometry}
\usepackage{xcolor}
\usepackage{hyperref}
\usepackage{xcolor}
\usepackage{multirow}
\usepackage{tabularx,booktabs}
\usepackage{lipsum}
\usepackage{diagbox}
\usepackage{tablefootnote}

\newtheorem{definition}{Definition}
\newtheorem{theorem}{Theorem}
\newtheorem{example}{Example}

\newtheorem{lemma}{Lemma}
\newtheorem{remark}{Remark}

\newtheorem{construction}{Construction}

\def\qi#1 {\fbox {\footnote {\ }}\ \footnotetext { From Qi: {\color{red}#1}}}

\newcolumntype{C}{>{\centering\arraybackslash}X} 
\title{On the cross-correlation properties of large-size families of Costas arrays}

\begin{document}

\author{Runfeng~Liu, \and Qi~Wang
\thanks{
Runfeng Liu is with the Department of Computer Science and Engineering, Southern University of Science and Technology, Shenzhen 518055, China. E-mail: 12332425@mail.sustech.edu.cn.
}
\thanks{
Q. Wang is with the Department of Computer Science and Engineering, and also with National Center for Applied Mathematics Shenzhen, Southern University of Science and Technology, Shenzhen 518055, China. E-mail: wangqi@sustech.edu.cn.
}

}

\maketitle

\begin{abstract}
Costas arrays have been an interesting combinatorial object for decades because of their optimal aperiodic auto-correlation properties. Meanwhile, it is interesting to find families of Costas arrays or extended arrays with small maximal cross-correlation values, since for applications in multi-user systems, the cross-interferences between different signals should also be small. The objective of this paper is to study several large-size families of Costas arrays or extended arrays, and their values of maximal cross-correlation are partially bounded for some cases of horizontal shifts $u$ and vertical shifts $v$. Given a prime $p \geq 5$, a large-size family of Costas arrays over $\{1, \ldots, p-1\}$ is investigated, including both the exponential and logarithmic Welch Costas arrays. An upper bound
on the maximal cross-correlation of this family for arbitrary $u$ and $v$ is given. We also show that the maximal cross-correlation of the family of power permutations over $\{1, \ldots, p-1\}$ for $u = 0$ and $v \neq 0$ is bounded by $\frac{1}{2}+\sqrt{p-1}$. Furthermore, we give the first nontrivial upper bound on the maximal cross-correlation of the larger family including both exponential Welch Costas arrays and power permutations over $\{1, \ldots, p-1\}$ for arbitrary $u$ and $v=0$ that it equals $(p-1)/t$ where $t$ is the smallest prime divisor of $(p-1)/2$ if p is not a safe prime and is at most $(p-1)^{\frac{1}{2}}+(p-1)^{\frac{1}{4}}+\frac{1}{2}$ otherwise.
\end{abstract}


\section{Introduction}\label{sec-intro}

Costas arrays, first introduced by Costas in 1965~\cite{Costas1965Sonar}, have attracted much attention for decades due to their optimal aperiodic auto-correlation properties. A Costas array of order $n$ is a permutation matrix of dots/1's and blanks/0's such that there exists exactly one dot in each row and each column, respectively, and its aperiodic auto-correlation values, except for the trivial case, are less than or equal to $1$. The aperiodic auto-correlation value, at the horizontal shift $u$ and the vertical shift $v$, is the number of coincidences between dots in the permutation matrix and dots in its corresponding shifted version. Costas arrays play a crucial role in the detection of both the position and the velocity of target objects in radar and sonar systems, as they possess such an optimal ``thumbtack'' ambiguity function of auto-correlation values~\cite{GT1982Two},\cite{Costas1984study}.

While the auto-correlation sidelobe values of Costas arrays are always no greater than $1$, Costas arrays in general perform poorly in terms of their cross-correlation properties. In 1985, Freedman and Levanon~\cite{FL1985Two} proved that the maximal cross-correlation of every pair of different Costas arrays of order $n \geq 4$ is at least $2$. Despite this fact, it is still interesting to find families of Costas arrays with small maximal cross-correlation values, since in particular for applications in multi-user systems, the cross-interferences between different signals should be restricted to a level much smaller than the order of the arrays~\cite{DT1991Cross-correlation}. Note that it remains a longstanding open problem whether Costas arrays exist for an arbitrary order $n$~\cite{GT1984Constructions},\cite{WCS2023} and most of Costas arrays are sporadic ones which cannot be constructed systematically~\cite{AP2022Transformation,TR2011Matlab}. The two widely investigated families of Costas arrays are the family of Welch Costas arrays and the family of Lempel-Golomb Costas arrays~\cite{DG2011Maximal,GW2020Note,LW2024cross-correlation}.

Recently, there have been some new results on the two algebraically constructed families of Costas arrays. For the family of exponential Welch Costas arrays, denoted by $\mathcal{W}_p$, Drakakis, Gow, Rickard, Sheekey and Taylor~\cite{DG2011Maximal} showed that its maximal cross-correlation, for arbitrary $u$ and $v=0$, equals $\frac{p-1}{t}$, where $p$ is a prime and $t$ is the smallest prime divisor of $\frac{p-1}{2}$. For the case that $u$ is arbitrary and $v \neq 0$, G\'omez-P\'erez and Winterhof~\cite{GW2020Note} further derived the upper bound of $1+\left\lfloor\left(1-\frac{2}{p-1}\right) p^{1 / 2}\right\rfloor$, using the tools of character theory and the Weil bound. These results are summarized in Table~\ref{table_Wp_Wpel}. The recent progress on the maximal cross-correlation of the family of Lempel-Golomb Costas arrays can be found for example in~\cite{DG2011Maximal,GW2020Note,LW2024cross-correlation}.

While studying the differential property of power functions, Ardalani~\cite{Ardalani2023Contributions} observed an interesting relation between the maximal cross-correlation of the family of exponential Welch Costas arrays and the differential property of the family of power permutations, denoted by $\mathcal{P}_p$. As Ardalani thereby defined the cross-correlation of the family of power permutations, the maximal cross-correlation of the family $\mathcal{P}_p$ was partially bounded. More precisely, Ardalani proved that the maximal cross-correlation of power permutations equals $\frac{p-1}{t}$ for the special case $u = v = 0$, where $t$ is the smallest prime divisor of $\frac{p-1}{2}$, and is bounded by $\left\lceil\frac{p-2}{p-1}(1+\sqrt{p})\right\rceil$ for $v=0$ and $u \neq 0$. The remaining two cases $u = 0, v \neq 0$ and $u \neq 0, v \neq 0$ are left as open problems. Ardalani further extended the original family of exponential Welch Costas arrays $\mathcal{W}_p$ by including the power permutations, denoted by $\mathcal{PW}_p$, and examined its cross-correlation properties. However, except some numerical results, generic theoretical bounds on the maximal cross-correlation of this extended family are missing.

The objective of this paper is to establish certain theoretical bounds on the maximal cross-correlation of several large-size families of Costas arrays or extended arrays, and the main contributions are summarized as follows. First, a large-size family of Costas arrays, denoted by $\mathcal{W}_p^{el}$, is introduced by combining together the family of exponential Welch Costas arrays and that of logarithmic Welch Costas arrays. An upper bound on the maximal cross-correlation of this large family for arbitrary $u$ and $v$ is derived, using the idea that transforms the transcendental equations into polynomial functions~\cite{GS2010Cycles} (see Table~\ref{table_Wp_Wpel}). We then revisit the family $\mathcal{PW}_p$, and give the first nontrivial upper bound of its maximal cross-correlation for arbitrary $u$ and $v = 0$. The main tool is to transform the problem of calculating the cross-correlation values into that of counting the number of fixed points of Welch Costas arrays. Meanwhile, we are able to derive an upper bound on the maximal cross-correlation of the family $\mathcal{P}_p$ for $u = 0$ and $v \ne 0$. See Table~\ref{table_Pp_PWp} for relevant bounds on certain cases of these large-size families of extended arrays.

\begin{table}[h]
\centering
\begin{tabular}{ccccc}
\toprule
Family   &  Family size & Bound & $(u,v)$ & Reference \\
\midrule
\multirow{2}{*}{$\mathcal{W}_p$} &  \multirow{2}{*}{$\phi(p-1)$} & $\frac{p-1}{t}$ & arbitrary u, $v=0$ & \cite{DG2011Maximal} \\
& & $1+\left\lfloor\left(1-\frac{2}{p-1}\right) p^{\frac{1}{2}}\right\rfloor$ & arbitrary u, $v \neq 0$ & \cite{GW2020Note}\\
\addlinespace
\multirow{2}{*}{$\mathcal{W}_p^{el}$} &  \multirow{2}{*}{$2\phi(p-1)$} & $\max \left\{4 p \log _p \alpha+1,1+\left\lfloor\left(1-\frac{2}{p-1}\right) p^{\frac{1}{2}}\right\rfloor\right\}$, \text{ if} $p$ \text{is a safe prime }& \multirow{2}{*}{arbitrary $u$, $v$} & \multirow{2}{*}{This paper} \\
& & $\max \left\{4 p \log _p \alpha+1, \frac{p-1}{t}\right\}$, \text{ if} $p$ \text{is not a safe prime } &  & \\

\bottomrule
\end{tabular}
\caption{Bounds on the maximal cross-correlation of $\mathcal{W}_p$ and $\mathcal{W}_p^{el}$}
\label{table_Wp_Wpel}
\end{table}

\begin{table}[h]
\centering
\begin{tabular}{ccccc}
\toprule
Family   &  Family size & Bound & $(u,v)$ & Reference \\
\midrule
\multirow{4}{*}{$\mathcal{P}_p$} & \multirow{4}{*}{$\phi(p-1) - 1$} & $\frac{p-1}{t}$ & $u=0, v=0$ & \cite{Ardalani2023Contributions} \\
& & $\frac{p-2}{p-1}(1+\sqrt{p})$ & $u \neq 0, v = 0$ & \cite{Ardalani2023Contributions} \\
& & $\frac{1}{2}+\sqrt{p-1}$ & $u = 0, v \neq 0$ & This paper \\
& & \text{ unknown } & $u \neq 0, v \neq 0$ & Open \\
\addlinespace
\multirow{3}{*}{$\mathcal{PW}_p$}& \multirow{3}{*}{$2\phi(p-1) - 1$}  & $(p-1)^{\frac{1}{2}}+(p-1)^{\frac{1}{4}}+\frac{1}{2}$, \text{ if} $p$ \text{is a safe prime}& \multirow{2}{*}{arbitrary $u$, $v = 0$} & \multirow{2}{*}{This paper} \\
& & $\frac{p-1}{t}$, \text{ if} $p$ \text{is not a safe prime} &  & \\
& & \text{ unknown } & arbitrary $u$, $v \neq 0$ & Open \\
\bottomrule
\end{tabular}
\caption{Bounds on the maximal cross-correlation of $\mathcal{P}_p$ and $\mathcal{PW}_p$}
\label{table_Pp_PWp}
\end{table}

The rest of this paper is organized as follows. In Section~\ref{sec-pre}, we provide some necessary definitions and preliminaries that will be useful in the subsequent sections. In Section~\ref{sec-wp}, we introduce a new family that includes both the exponential and the logarithmic Welch Costas arrays and study its maximal cross-correlation. In Section~\ref{sec-pwp}, we obtain certain upper bounds on the maximal cross-correlation of $\mathcal{P}_p$ and $\mathcal{PW}_p$
for some cases, respectively. Finally, we conclude this paper with some open problems in Section~\ref{sec-con}.

\section{Preliminaries}\label{sec-pre}

Let $n$ be a positive integer and $[n]$ denote the set $\{1, 2, \ldots, n\}$. Let $u, v$ be two integers, which typically denote the horizontal and the vertical shifts of a permutation matrix, respectively. A bijection $f$ from $[n]$ to $[n]$ is called a permutation of order $n$, denoted as $f= [f(i)]=[f(1), f(2), \ldots, f(n)]$. The corresponding permutation matrix is denoted by $A_f=(a_{ij})_{n \times n}$ where $a_{ij} = 1$ if and only if $f(j) = i$ and $a_{ij}=0$ otherwise. When no ambiguity arises, we omit subscript $f$. This permutation matrix can be regarded as an equivalent representation of the permutation $f$.

\begin{definition}
Let $f, g: [n] \rightarrow[n]$ be two permutations of order $n$ and let $A, B$ be their corresponding permutation matrices, respectively. The cross-correlation between $f$ and $g$ at $(u, v)$ is defined as
$$
\Psi_{A, B}(u, v) = |\{(f(i)+v, i+u): i \in[n]\} \cap\{(g(i), i): i \in[n]\}|.
$$
\end{definition}

In particular, the cross-correlation is the aperiodic auto-correlation when $f=g$. In the following, we give the formal definitions of Costas arrays through the property of the aperiodic auto-correlation.

\begin{definition}
Let $f= [f(i)] = [f(1), f(2), \ldots, f(n)]: [n] \rightarrow [n]$ be a permutation of order $n$, and let $A = (a_{ij})_{n \times n}$ be the corresponding permutation matrix. Then $A$ is a Costas array if the aperiodic auto-correlation of A satisfies
$\Psi_{A, A}(u, v) \leq 1$ for all $(u, v) \neq(0,0)$.
\end{definition}

\begin{example}
Consider the permutation $f = [3,2,6,4,5,1]$ and represent the permutation $f$ in the matrix form, denoted by $A$, as
\begin{equation*}
A =
\left[\begin{array}{llllll}
0 & 0 & 0 & 0 & 0 & 1 \\
0 & 1 & 0 & 0 & 0 & 0 \\
1 & 0 & 0 & 0 & 0 & 0 \\
0 & 0 & 0 & 1 & 0 & 0 \\
0 & 0 & 0 & 0 & 1 & 0 \\
0 & 0 & 1 & 0 & 0 & 0 \\
\end{array}
\right].
\end{equation*}

For $-5 \le u,v \le 5$, we calculate the auto-correlation value $\Psi_{A, A}(u, v)$ of $A$ at the shift $(u,v)$ relative to itself. The results are recorded in the matrix $\mathcal{C}_{A}$ as follows, with the $(0,0)$ shift positioned centrally.

\begin{equation*}
\mathcal{C}_{A} = \left[\begin{array}{lllllllllll}
0 & 0 & 0 & 0 & 0 & 0 & 0 & 1 & 0 & 0 & 0 \\
0 & 0 & 0 & 1 & 0 & 0 & 1 & 0 & 0 & 0 & 0 \\
0 & 0 & 1 & 0 & 1 & 0 & 0 & 0 & 0 & 0 & 1 \\
0 & 0 & 1 & 1 & 0 & 0 & 1 & 0 & 1 & 0 & 0 \\
0 & 1 & 0 & 0 & 1 & 0 & 1 & 1 & 0 & 1 & 0 \\
0 & 0 & 0 & 0 & 0 & 6 & 0 & 0 & 0 & 0 & 0 \\
0 & 1 & 0 & 1 & 1 & 0 & 1 & 0 & 0 & 1 & 0 \\
0 & 0 & 1 & 0 & 1 & 0 & 0 & 1 & 1 & 0 & 0 \\
1 & 0 & 0 & 0 & 0 & 0 & 1 & 0 & 1 & 0 & 0 \\
0 & 0 & 0 & 0 & 1 & 0 & 0 & 1 & 0 & 0 & 0 \\
0 & 0 & 0 & 1 & 0 & 0 & 0 & 0 & 0 & 0 & 0 \\
\end{array}
\right].
\end{equation*}

We can easily observe that the values of auto-correlation for all shifts $(u,v)$ are no greater than $1$ except for the trivial case $(u, v)=(0,0)$, which implies that the matrix $A$ is a Costas array of order $6$.
\end{example}

The definition above of Costas arrays is via aperiodic auto-correlation and the equivalent two definitions of Costas arrays via distance vectors and difference triangles can be found for example in~\cite{GT1984Constructions, Golomb1984Algebraic}.

Below we recall the classical algebraic construction of Costas arrays by Welch~\cite{Golomb1984Algebraic}, and the other classical Lempel-Golomb construction can be found in~\cite{Golomb1984Algebraic}.

\begin{construction}\label{const-exp-welch}\cite{GT1984Constructions}
  {\em [Exponential Welch Construction]}
Denote by $\mathbb{F}_p$ the finite field of order $p$, where $p>2$ is a prime. Let $\alpha$ be a primitive element of $\mathbb{F}_p$ and $c$ be an integer with $0 \leq c \leq p-2$.
The $(p-1) \times( p-1)$ permutation matrix, denoted by $W_1^{\exp}(p, \alpha, c) = (a_{ij})_{(p-1) \times (p-1)}$, where $a_{i j}=1$ if and only if $i \equiv \alpha^{j-1+c}( \bmod \, p)$, for $1 \leq i, j \leq p-1$, is a Costas array.
\end{construction}

The following logarithmic Welch construction is the inverse permutation of the exponential Welch construction.

\begin{construction}\label{const-log-welch}{\em [Logarithmic Welch Construction]}
Denote by $\mathbb{F}_p$ the finite field of order $p$, where $p>2$ is a prime. Let $\alpha$ be a primitive element of $\mathbb{F}_p$ and $c$ be an integer with $0 \leq c \leq p-2$.
The $(p-1)\times (p-1)$ permutation matrix, denoted by $W_1^{\log}(p, \alpha, c) = (a_{ij})_{(p-1) \times (p-1)}$, where $a_{ij} = 1$ if and only if $i \equiv 1+c+\log _\alpha j \bmod (p-1)$, for $1 \leq i, j\leq p-1$, is a Costas array.
\end{construction}

In a Costas array, the points on the main diagonal are called {\em fixed points}. The number of fixed points plays an important role in understanding the structure of the algebraically constructed Costas arrays~\cite{Drakakis2011Open,JW2014Deficiency,JY2018cubes}. In fact, the fixed points in a Costas array form a {\em Golomb ruler} (or called {\em Sidon set}, {\em $B_2$-sequence}). Formally, a {\em Golomb ruler} is a set of distinct integers such that all differences between any two integers are different. Therefore, the number of the fixed points in Costas arrays of order $n$ is naturally bounded by the size of the Golomb ruler in the set $[n]$. Cilleruelo~\cite{Cilleruelo2010sidon} gave the following upper bound on the size of Golomb rulers while studying the size of Sidon sets.

\begin{lemma}\cite[Corollary 2]{Cilleruelo2010sidon}
\label{size_sidon}
Let $F(n)$ denote the maximal cardinality of a Sidon set in $[n]$, then
$$
F(n) < n^{\frac{1}{2}}+n^{\frac{1}{4}}+ \frac{1}{2}.
$$
\end{lemma}

Note that the cross-correlation between a Costas array constructed by the exponential Welch construction (Construction~\ref{const-exp-welch}) and its cyclically shifted version may contain values between $1$ and $p-2$, where the cyclically shifted version means a Costas array by the exponential Welch construction using the same primitive element but a different constant $c$. The family of exponential Welch Costas arrays, denoted by $\mathcal{W}_p$, does not contain cyclically shifted Costas arrays. In other words, the constant $c$ is fixed and set to $0$ for convenience to generate a family of Costas arrays by Construction~\ref{const-exp-welch}. In a similar way, the family of logarithmic Welch Costas arrays, denoted by $\mathcal{W}_p^l$, is also defined in the following.


\begin{construction}\label{const-wp}
For prime $p \geq 5$, the family of exponential Welch Costas arrays of order $(p-1)$ that does not contain any cyclic shift of a given exponential Welch array $W_1^{\exp}(p, \alpha, 0)$, is defined as
$$
\mathcal{W}_p=\left\{\left[\alpha^{i-1} \bmod p\right]: 1 \leq i \leq p-1 \text { and } \alpha \text { is a primitive element of } \mathbb{F}_p\right\}.
$$
Note that $\left|\mathcal{W}_p\right|=\phi(p-1)$, where $\phi$ is the Euler's totient function.
\end{construction}

\begin{construction}\label{const-wpl}
For prime $p \geq 5$, the family of logarithmic Welch Costas arrays of order $(p-1)$ that does not contain any cyclic shift of a given logarithmic Welch array $W_1^{\log}(p, \alpha, 0)$, is defined as
$$
\mathcal{W}^{l}_p=\left\{\left[ (1+\log _\alpha i) \bmod (p-1) \right]: 1 \leq i \leq p-1 \text { and } \alpha \text { is a primitive element of } \mathbb{F}_p\right\}.
$$
Clearly, the family $\mathcal{W}_p^l$ has the same size $\phi(p-1)$ as the family $\mathcal{W}_p$.
\end{construction}

To derive theoretical bounds on the maximal cross-correlation of Costas arrays, the following classical Weil bound for multiplicative character sums is very important.

\begin{lemma}\cite[Theorem 5.41]{LN1997finitefield}\label{lem-weilbound}
Let $\mathbb{F}_p$ be a finite field with $p$ elements and let $\psi$ be a multiplicative character of $\mathbb{F}_p$ of order $m>1$. Suppose that $f \in \mathbb{F}_p[x]$ has precisely $d$ distinct roots in its splitting field over $\mathbb{F}_p$, and suppose that $f$ is not an m-th power of a polynomial. Then for every $a \in \mathbb{F}_p$, we have
$$
\left|\sum_{x \in \mathbb{F}_p} \psi(a f(x))\right| \leq(d-1) \sqrt{p}.
$$

\end{lemma}

The Weil bound above becomes trivial in some cases, for example, when the number of distinct roots $d$ is large. Kelley and Owen~\cite{KO2017trinomials} gave the following upper bound on the number of roots of trinomials over finite fields, which will be used in deriving bounds on the maximal cross-correlation of Costas arrays afterwards.


\begin{lemma}\cite[Theorem 1.2]{KO2017trinomials}
\label{lem-kelley}
For a trinomial
$$
f(x)=x^n+a x^s+b \in \mathbb{F}_p[x],
$$
with $a$ and $b$ both nonzero, the number of roots $R(f) \leq \delta \left\lfloor\frac{1}{2}+\sqrt{\frac{p-1}{\delta}}\right\rfloor$, where $\delta=\operatorname{gcd}(n, s, p-1)$, and $\lfloor \cdot \rfloor$ denotes the floor function.
\end{lemma}

The family of power permutations $\mathcal{P}_p$, studied by Ardalani~\cite{Ardalani2023Contributions}, is defined as follows.

\begin{construction}\label{const-pp}
For prime $p \geq 5$, the family $\mathcal{P}_p$ of power permutations of order $(p-1)$, is defined as
$$
\mathcal{P}_p=\left\{\left[i^d \bmod p\right]: 1 \leq i \leq p-1 \text { where } \operatorname{gcd}(d, p-1)=1 \text{ and } 1 < d \le p-2\right\}.
$$
Note that $\left|\mathcal{P}_p\right|=\phi(p-1)-1$ where the subtraction of $1$ corresponds to excluding the identity permutation $d=1$ and $\phi$ denotes Euler's totient function.
\end{construction}

\section{The family of Costas arrays $\mathcal{W}^{el}_p$ and its cross-correlation}
\label{sec-wp}

The main goal of this section is to extend the family of exponential Welch Costas arrays $\mathcal{W}_p$ by including the family of logarithmic Welch Costas arrays and study the cross-correlation properties of this larger family since the logarithmic Welch Costas arrays by Construction~\ref{const-log-welch} are distinct from those by Construction~\ref{const-exp-welch} for $p > 5$~\cite{DG2009Symmetry}. To best of our knowledge, this new family has not been investigated before. Note that the size of this family is twice as large as the family $\mathcal{W}_p$ which has been widely studied.

\begin{construction}\label{const-wpel}
For $p > 5$, the family $\mathcal{W}^{el}_p$ of Costas arrays of order $(p-1)$ is defined as
$$
\mathcal{W}^{el}_p=\mathcal{W}_p \bigcup \mathcal{W}^{l}_p ,
$$
where $\mathcal{W}_p$ and $\mathcal{W}_p^l$ are defined in Construction~\ref{const-wp} and Construction~\ref{const-wpl}, respectively. Note that $|\mathcal{W}^{el}_p|=2\phi(p-1)$, where $\phi$ denotes Euler's totient function.
\end{construction}

As shown in Table~\ref{cross_welch}, we list exhaustively the maximal cross-correlation values of two Costas arrays in the family $\mathcal{W}_p^{el}$, i.e., $\mathcal{C}(\mathcal{W}^{el}_p) = \max \limits_{(u, v)} \max \limits_{f \neq g} \Psi_{f, g}(u, v)$, where $f, g \in \mathcal{W}_p^{e l}$ and $0 \leq |u|, |v| \leq p-2$, for $7 \leq p \leq 277$. The maximal cross-correlation values of the family $\mathcal{W}_p$ are also included in Table~\ref{cross_welch} for comparison, denoted by $\mathcal{C}(\mathcal{W}_p)$, which can also be found in~\cite{DG2011Maximal}.

\begin{table}[h]
\centering
\resizebox{\textwidth}{!}{
\begin{tabular}{|ccc|ccc|ccc|ccc|}
\hline
$p$ & $\mathcal{C}(\mathcal{W}_p)$ & $\mathcal{C}(\mathcal{W}^{el}_p)$ & $p$ & $\mathcal{C}(\mathcal{W}_p)$ & $\mathcal{C}(\mathcal{W}^{el}_p)$ & $p$ & $\mathcal{C}(\mathcal{W}_p)$ & $\mathcal{C}(\mathcal{W}^{el}_p)$ & $p$ & $\mathcal{C}(\mathcal{W}_p)$ & $\mathcal{C}(\mathcal{W}^{el}_p)$ \\
\hline
\underline{7} & \underline{2} & \underline{3} & 61  & 30 & 30 & 131 & 26 & 26 & 199 & 66  & 66  \\
\underline{11} & \underline{3} & \underline{5} & 67  & 22 & 22 & 137 & 68 & 68 & 211 & 70  & 70  \\
13 & 6  & 6  & 71  & 14 & 14 & 139 & 46 & 46 & 223 & 74  & 74  \\
17 & 8  & 8  & 73  & 36 & 36 & 149 & 74 & 74 & \underline{227} & \underline{6} & \underline{11} \\
19 & 6  & 6  & 79  & 26 & 26 & 151 & 50 & 50 & 229 & 114 & 114 \\
\underline{23} & \underline{4} & \underline{7} & \underline{83} & \underline{5} & \underline{8} & 157 & 78 & 78 & 233 & 116 & 116 \\
29 & 14 & 14 & 89  & 44 & 44 & 163 & 54 & 54 & 239 & 34  & 34  \\
31 & 10 & 10 & 97  & 48 & 48 & \underline{167} & \underline{6} & \underline{11} & 241 & 120 & 120 \\
37 & 18 & 18 & 101 & 50 & 50 & 173 & 86 & 86 & 251 & 50  & 50  \\
41 & 20 & 20 & 103 & 34 & 34 & \underline{179} & \underline{6} & \underline{10} & 257 & 128 & 128 \\
43 & 14 & 14 & \underline{107} & \underline{5} & \underline{8} & 181 & 90 & 90 & \underline{263} & \underline{7} & \underline{11} \\
\underline{47} & \underline{5} & \underline{9} & 109 & 54 & 54 & 191 & 38 & 38 & 269 & 134 & 134 \\
53 & 26 & 26 & 113 & 56 & 56 & 193 & 96 & 96 & 271 & 90  & 90  \\
\underline{59} & \underline{5} & \underline{9} & 127 & 42 & 42 & 197 & 98 & 98 &  277   &   138  &   138 \\
\hline
\end{tabular}
}
\caption{Maximal cross-correlation of $\mathcal{W}_p$ and $\mathcal{W}^{el}_p$}
\label{cross_welch}
\end{table}

Recall that the maximal cross-correlation of $\mathcal{W}_p$, for arbitrary $u$ and $v = 0$, is equal to $\frac{p-1}{t}$, where $t$ is the smallest prime divisor of $\frac{p-1}{2}$ and, for arbitrary $u$ and $v \neq 0$, is bounded by $1+\left\lfloor\left(1-\frac{2}{p-1}\right) p^{1 / 2}\right\rfloor$. It then follows that when $p$ is not a safe prime, $\mathcal{C}(\mathcal{W}_p) = \frac{p-1}{t}$, since $1+\left\lfloor\left(1-\frac{2}{p-1}\right) p^{1 / 2}\right\rfloor \leq \frac{p-1}{t}$ and when $p$ is a safe prime, $\mathcal{C}(\mathcal{W}_p) \leq 1+\left\lfloor\left(1-\frac{2}{p-1}\right) p^{1 / 2}\right\rfloor$, since $\frac{p-1}{t} = 2$. Here, a prime $p$ is a safe prime if and only if $p=2r+1$ and $r$ is also a prime (also known as {\em Sophie Germain prime}).

It is observed that there are also two cases on the maximal cross-correlation of the family $\mathcal{W}_p^{el}$:
(1) When $p$ is not a safe prime, see $p$ values with no underline in Table~\ref{cross_welch}, $\mathcal{C}(\mathcal{W}^{el}_p) = \mathcal{C}(\mathcal{W}_p) = \frac{p-1}{t}$, where $t$ is the least prime divisor of $\frac{p-1}{2}$. For example, let $p=103$ and factorize $p-1$ to obtain $103 - 1 = 2 \times 3 \times 17$, writing the factors in the ascending order. The least prime divisor of $\frac{p-1}{2}$ is equal to $3$, then $\frac{p-1}{t} = \frac{102}{3} = 34$, corresponding to the result in the column $\mathcal{C}(\mathcal{W}^{el}_p)$; (2) When $p$ is a safe prime, see underline $p$ values in Table~\ref{cross_welch}, $\mathcal{C}(\mathcal{W}^{el}_p) \ne \mathcal{C}(\mathcal{W}_p)$ and $\mathcal{C}(\mathcal{W}^{el}_p), \mathcal{C}(\mathcal{W}_p)$ are both small compared to the order of the arrays.

To prove some of the above observations, we need to calculate the cross-correlation between any two Costas arrays $f,g$ in $\mathcal{W}^{el}_p$ at the horizontal shift $u$ and the vertical shift $v$. There are three cases where $f$ and $g$ are both exponential Welch Costas arrays, both logarithmic Welch Costas arrays or one is an exponential Welch Costas array while the other is a logarithmic Welch Costas array.
Since the maximal cross-correlation between either two exponential Welch Costas arrays or two logarithmic Welch Costas arrays already have settled bounds~\cite{GW2020Note}, the remaining question is to compute the cross-correlation between an exponential Costas array and a logarithmic one at the shift $(u,v)$.

Consider an exponential Welch Costas array $f$ and a logarithmic Welch Costas array $g$, generated by primitive elements $\alpha$ and $\beta$, respectively. Without loss of generality, let the coordinates of $f$ be $(i_1,j_1)$ where $i_1 \equiv \alpha^{j_1-1}( \bmod \, p)$ and those of $g$ be $(i_2,j_2)$ where $i_2 \equiv 1+\log _\beta j_2 \bmod (p-1)$. It is equivalent to represent the coordinates of $g$ as $(j_2,i_2)$ where $i_2 \equiv \beta^{j_2-1}( \bmod \, p)$. In order to derive generic equations for calculating the cross-correlation between the two Costas arrays at $(u,v)$, we need to solve the following system:

\begin{equation}
\label{eq_system}
\left\{
\begin{aligned}
& j_2=i_1+v , \\
& i_2=j_1+u . \\
\end{aligned}
\right.
\end{equation}

By substituting $i_1 \equiv \alpha^{j_1-1}( \bmod \, p)$ and $i_2 \equiv \beta^{j_2-1}( \bmod \, p)$ into Eq.~(\ref{eq_system}), we have
\begin{equation*}
\left\{
\begin{aligned}
& j_2=\alpha^{j_1-1}( \bmod \, p)+v ,  \\
& \beta^{j_2-1}( \bmod \, p)=j_1+u . \\
\end{aligned}
\right.
\end{equation*}

Let $x=j_1$ and $y=j_2$, and the value of cross-correlation $\Psi_{f, g}(u, v)$ is then equal to the number of solution pairs $(x,y)$ satisfying the following equations:

\begin{equation}
\label{equation_of_Wp^el}
 \left\{
\begin{aligned}
& y = \alpha^{x-1} \bmod p + v , \\
& \beta^{y-1} \bmod p = x + u . \\
\end{aligned}
\right.
\end{equation}


It seems hard to directly determine the number of solutions of Eq.~(\ref{equation_of_Wp^el}) because this is a transcendental equation system over a finite field. Next, we give an upper bound of the number of solutions of Eq.~(\ref{equation_of_Wp^el}) for arbitrary $u$ and $v$, using the idea from~\cite{GS2010Cycles} that simplifies the equation system into a polynomial equations.

\begin{lemma}
\label{lemma-Wpel}
Let $f$ be an exponential Welch Costas array $W_1^{\exp}(p, \alpha, 0)$ and $g$ be a logarithmic Welch Costas array $W_1^{\log}(p, \beta, 0)$ over $\mathbb{F}_p$ with prime $p$ and $\alpha, \beta \in \{1,2,\cdots,p-1\}$ two primitive elements. Then for arbitrary $u$ and $v$, we have
$$
\Psi_{f, g}(u,v) \leq 4 p \, \log_p \alpha  + 1,
$$
where $\log_p \alpha$ is the logarithmic function on real numbers.
\end{lemma}

\begin{proof}
Suppose that $(x_1,y_1),(x_2,y_2),\cdots,(x_N,y_N)$ are the solutions of Eq.~(\ref{equation_of_Wp^el}) where $0 \le x_i,y_i \le p-1$ for $i=1,2,\cdots,N$. Then, for $i=1,2,\cdots,N$, we have
\begin{equation*}
\left\{
\begin{aligned}
& y_i = \alpha^{x_i-1} \bmod p + v , \\
& \beta^{y_i-1} \bmod p = x_i + u . \\
\end{aligned}
\right.
\end{equation*}

Without loss of generality, assume that $x_1 < x_2 < \cdots < x_N$.


Choose a positive parameter $z<p$ to be optimized later and consider the number of indices $i \in\{1,2, \ldots, N-1\}$ satisfying $x_{i+1}-x_i \geq z$. Including the last index $i=N$, this number, denoted by $I_0$, is at most $\frac{p}{z}+1$.

For each positive integer $a<z$, let $I_a$ denote the set of indices $i \in\{1,2, \ldots, N-1\}$ such that $x_{i+1}-x_i=a$. The values $x_i$ and $x_{i+1}$ that satisfy this condition must also fulfill
\begin{equation*}
\left\{
\begin{aligned}
& \beta^{y_i-1} \bmod p = x_i + u, \\
& \beta^{y_{i+1}-1} \bmod p = x_{i+1} + u. \\
\end{aligned}
\right.
\end{equation*}

After substituting $x_{i+1}-x_i=a$, the equations become
$$
\beta^{y_i-1} \bmod p + a = \beta^{y_{i+1}-1} \bmod p.
$$

Applying the modulus $p$ to both sides yields
\begin{equation}
\label{eq_beta}
\beta^{y_i-1} + a \equiv \beta^{y_{i+1}-1} \bmod p.
\end{equation}

The corresponding $y_i$ and $y_{i+1}$ satisfy
\begin{equation*}
\left\{
\begin{aligned}
& y_i = \alpha^{x_i-1} \bmod p + v, \\
& y_{i+1} = \alpha^{x_{i+1}-1} \bmod p + v. \\
\end{aligned}
\right.
\end{equation*}

Substituting $x_{i+1}-x_i=a$ into the equations, we obtain
$$
y_{i+1} - v = \alpha^a \alpha^{x_i-1} \bmod p.
$$

Because $y_i-v=\alpha^{x_i-1} \bmod p$, it follows that
$$
y_{i+1} - v = \alpha^a (y_i - v)\bmod p.
$$

Noting that $y_i-v, y_{i+1} - v \in \{0,1, \ldots, p-1\}$ and using ~\cite[Lemma 1]{GS2010Cycles}, we have
\begin{equation}
\label{eq-lemma}
\beta^{y_{i+1} - v} \equiv \beta^{\alpha^{a}(y_i-v) - k} \bmod p ,
\end{equation}
where $k=\left\lfloor\frac{\alpha^{a}(y_i-v)}{p}\right\rfloor<\alpha^{a}$.

Thus, by combining Eq.~(\ref{eq_beta}) and Eq.~(\ref{eq-lemma}), we obtain
$$
\beta^{\alpha^{a}(y_i-v)-k} \equiv \beta^{y_i-v} + \beta^{-v+1} a  \bmod p.
$$

Let $t = \beta^{y_i-v}$. It then follows that
\begin{equation}
\label{equationWelch}
\beta^{-k} t^{\alpha^{a}} \equiv t + \beta^{-v+1} a \bmod p.
\end{equation}

Upon obtaining the solution $t$ of the equation, a unique $y_i$ can be identified, which subsequently allows for the determination of the corresponding $x_i$. Observing that there are at most $\alpha^a$ choices for $k$, we obtain the upper bound of $I_a$ from Lemma~\ref{lem-kelley}
$$
I_a \leqslant \alpha^{a} (\sqrt{p}+1), \quad a<z,
$$
since the Eq. (\ref{equationWelch}) is a trinomial polynomial over finite field. Recall that $z$ is a positive parameter to be chosen.

Therefore, we estimate the number of $N$ as
$$
N \leq I_0 + \sum_{a < z}I_a \leq \frac{p}{z} + 1 + \frac{\alpha^z-\alpha}{\alpha-1} (1+\sqrt{p}) \leq \frac{p}{z} + 1 + \alpha^z(1+\sqrt{p}).
$$

Taking $z$ as $\frac{1}{2} \log_\alpha p$ which is a logarithmic function on real numbers, we derive the upper bound on the number of solutions of Eq.~(\ref{equation_of_Wp^el})
$$
N \leq 4 p \, \log_p \alpha  + 1.
$$
\end{proof}

\begin{theorem}
\label{thm_Wpel}
For prime $p \geq 7$, the maximal cross-correlation $\mathcal{C}(\mathcal{W}_p^{e l})$ of the family of Welch Costas arrays $\mathcal{W}_p^{e l}$ over $\{1,2,\cdots,p-1\}$ satisfies

\begin{equation}
\mathcal{C}(\mathcal{W}^{el}_p) \leq \begin{cases} \max \left\{4 p \, \log_p \alpha  + 1, 1+\left\lfloor\left(1-\frac{2}{p-1}\right) p^{1 / 2}\right\rfloor\right\} , & \text { if } p \text { is a safe prime, } \\  \max \left\{4 p \, \log_p \alpha  + 1, \frac{p-1}{t}\right\} , & \text { otherwise, }\end{cases}
\tag{$\ast$}
\end{equation}
where $\alpha$ is the largest primitive element generating the exponential Costas array and $t$ is the smallest prime divisor of $\frac{p-1}{2}$.
\end{theorem}

\begin{proof}
This result directly follows from Lemma~\ref{lemma-Wpel} and \cite[Theorem 1]{GW2020Note}.
\end{proof}

\begin{remark}
Observing that the bound in Theorem~\ref{thm_Wpel} becomes trivial for large $\alpha$, we remark that this can be avoided by simply incorporating exponential Welch Costas arrays with small primitive elements. This approach generates a new family of arrays that combines the family of logarithmic Welch Costas arrays with exponential Welch Costas arrays constructed by small primitive elements. Note that the cardinality of this new family surpasses $\phi(p-1)$ yet remains upper bounded by $2 \phi(p-1)$.
\end{remark}

We compute the results of Theorem~\ref{thm_Wpel} using the least primitive element $\alpha$ in $\mathbb{F}_p$ for $7 \le p \le 277$. Specifically, we include the exponential Welch Costas array generated by the least primitive element $\alpha$ into the family of logarithmic Welch Costas arrays. We list the cases where the bound of Theorem~\ref{thm_Wpel} is nontrivial in Table~\ref{bound-Wpel}.



\begin{table}[h]
\centering
\begin{tabular}{|ccc|ccc|ccc|}
\hline
$p$ & $\alpha$ & Bound $(*)$ & $p$ & $\alpha$ & Bound $(*)$ & $p$ & $\alpha$ & Bound $(*)$\\
\hline
19  & 2 & 17 & 107 & 2 & 63  & 181 & 2 & 96  \\
29  & 2 & 23 & 113 & 3 & 105 & 197 & 2 & 103 \\
37  & 2 & 28 & 127 & 3 & 115 & 199 & 3 & 165 \\
53  & 2 & 37 & 131 & 2 & 74  & 211 & 2 & 109 \\
59  & 2 & 40 & 137 & 3 & 122 & 223 & 3 & 181 \\
61  & 2 & 41 & 139 & 2 & 78  & 227 & 2 & 116 \\
67  & 2 & 44 & 149 & 2 & 82  & 233 & 3 & 187 \\
83  & 2 & 52 & 163 & 2 & 55  & 257 & 3 & 203 \\
89  & 3 & 87 & 173 & 2 & 93  & 269 & 2 & 134 \\
101 & 2 & 60 & 179 & 2 & 95  &     &   &
\\
\hline
\end{tabular}
\caption{The bound of Theorem~\ref{thm_Wpel} for $7 \le p \le 277$}
\label{bound-Wpel}
\end{table}

\section{Bounds on the maximal cross-correlation of $\mathcal{P}_p$ and $\mathcal{PW}_p$}
\label{sec-pwp}

The extended family $\mathcal{PW}_p$ investigated by Ardalani~\cite[p.~117]{Ardalani2023Contributions}, which incorporates the family of exponential Welch Costas arrays and the family of power permutations, is defined as follows.

\begin{construction}\cite{Ardalani2023Contributions}
For prime $p \geq 5$, the family $\mathcal{PW}_p$ of order $p-1$, is defined as
$$
\mathcal{P W}_p=\mathcal{W}_p \bigcup \mathcal{P}_p,
$$
where $\mathcal{W}_p$ and $\mathcal{P}_p$ are defined in Construction~\ref{const-wp} and Construction~\ref{const-pp}, respectively. The size of the family $\mathcal{PW}_p$ equals $2\phi(p-1)-1$, where $\phi$ denotes Euler's totient function.
\end{construction}

For completeness, we give a simple lemma to explain why the intersection of $\mathcal{W}_p$ and $ \mathcal{P}_p$ is an empty set.
\begin{lemma}
$\mathcal{W}_p \bigcap \mathcal{P}_p = \emptyset$.
\end{lemma}

\begin{proof}
Suppose that there exists an integer $d$ with $\operatorname{gcd}(d, p-1)=1$, and a primitive element $\alpha$ satisfying $x^d \equiv \alpha^{x-1} \bmod p$ for all $x \in \mathbb{F}^*_p$. Then we have
$$
x^d \equiv \alpha^{x-1} \bmod p \text{ and } (x+1)^d \equiv \alpha^{x } \bmod p.
$$
By some substitutions, we have
$$
\alpha x^d \equiv (x+1)^d \bmod p ,
$$
which has at most $d$ distinct roots, making a contradiction.
\end{proof}

Ardalani~\cite{Ardalani2023Contributions} used exhaustive search to determine the maximal cross-correlation values of the family $\mathcal{P}_p$. The maximal cross-correlation value is defined as
$\mathcal{C}\left(\mathcal{P}_p\right)=\max \limits _{(u, v)} \max \limits _{f,g \in \mathcal{P}_p} \Psi_{f, g}(u, v)$
where $(u, v) \neq(0,0)$ if $f=g$ and $0 \leq |u|, |v| \leq p-2$ otherwise. Similarly, Ardalani~\cite{Ardalani2023Contributions} computed the maximal cross-correlation values for two elements from the family $\mathcal{PW}_p$,
i.e., $\mathcal{C}\left(\mathcal{PW}_p\right)=\max \limits _{(u, v)} \max \limits _{f,g \in \mathcal{PW}_p} \Psi_{f, g }(u, v)$, where $(u, v) \neq(0,0)$ if $f=g$ and $0 \leq |u|, |v| \leq p-2$ otherwise. These computations were performed for all primes $5 \leqslant p \leqslant 277$. The maximal cross-correlation values of the family $\mathcal{W}_p$ are included as well for comparison, see Table~\ref{PW}.

\begin{table}[h]
\centering
\resizebox{\textwidth}{!}{
\begin{tabular}{|cccc|cccc|cccc|}
\hline
$p$ & $\mathcal{C}(\mathcal{W}_p)$ & $\mathcal{C}(\mathcal{P}_p)$ & $\mathcal{C}(\mathcal{PW}_p)$  & $p$ & $\mathcal{C}(\mathcal{W}_p)$ & $\mathcal{C}(\mathcal{P}_p)$ & $\mathcal{C}(\mathcal{PW}_p)$ & $p$ & $\mathcal{C}(\mathcal{W}_p)$ & $\mathcal{C}(\mathcal{P}_p)$ & $\mathcal{C}(\mathcal{PW}_p)$ \\
\hline
\underline{5}  & \underline{2}  & \underline{2}  & \underline{3}  & 79  & 26 & 26 & 26 & \underline{179} & \underline{6}   & \underline{10}  & \underline{11}  \\
\underline{7}  & \underline{2}  & \underline{2}  & \underline{3}  & \underline{83}  & \underline{5}  & \underline{9}  & \underline{9}  & 181 & 90  & 90  & 90  \\
\underline{11} & \underline{3}  & \underline{3}  & \underline{4}  & 89  & 44 & 44 & 44 & 191 & 38  & 38  & 38  \\
13 & 6  & 6  & 6  & 97  & 48 & 48 & 48 & 193 & 96  & 96  & 96  \\
17 & 8  & 8  & 8  & 101 & 50 & 50 & 50 & 197 & 98  & 98  & 98  \\
19 & 6  & 6  & 6  & 103 & 34 & 34 & 34 & 199 & 66  & 66  & 66  \\
\underline{23} & \underline{4}  & \underline{6}  & \underline{6}  & \underline{107} & \underline{5}  & \underline{10} & \underline{10} & 211 & 70  & 70  & 70  \\
29 & 14 & 14 & 14 & 109 & 54 & 54 & 54 & 223 & 74  & 74  & 74  \\
31 & 10 & 10 & 10 & 113 & 56 & 56 & 56 & \underline{227} & \underline{6}   & \underline{10}  & \underline{10}  \\
37 & 18 & 18 & 18 & 127 & 42 & 42 & 42 & 229 & 114 & 114 & 114 \\
41 & 20 & 20 & 20 & 131 & 26 & 26 & 26 & 233 & 116 & 116 & 116 \\
43 & 14 & 14 & 14 & 137 & 68 & 68 & 68 & 239 & 34  & 34  & 34  \\
\underline{47} & \underline{5}  & \underline{8}  & \underline{8}  & 139 & 46 & 46 & 46 & 241 & 120 & 120 & 120 \\
53 & 26 & 26 & 26 & 149 & 74 & 74 & 74 & 251 & 50  & 50  & 50  \\
\underline{59} & \underline{5}  & \underline{12} & \underline{12} & 151 & 50 & 50 & 50 & 257 & 128 & 128 & 128 \\
61 & 30 & 30 & 30 & 157 & 78 & 78 & 78 & \underline{263} & \underline{7}   & \underline{12}  & \underline{12}  \\
67 & 22 & 22 & 22 & 163 & 54 & 54 & 54 & 269 & 134 & 134 & 134 \\
71 & 14 & 14 & 14 & \underline{167} & \underline{6}  & \underline{12} & \underline{12} & 271 & 90  & 90  & 90  \\
73 & 36 & 36 & 36 & 173 & 86 & 86 & 86 & 277 & 138 & 138 & 138 \\
\hline
\end{tabular}
}
\caption{Maximal Cross-correlation of $\mathcal{W}_p$, $\mathcal{P}_p$ and $\mathcal{PW}_p$~\cite{Ardalani2023Contributions}}
\label{PW}
\end{table}

For the maximal cross-correlation of the family of power permutations $\mathcal{P}_p$, Ardalani~\cite{Ardalani2023Contributions} settled the two cases when $u=0,v=0$ and $u \neq 0, v=0$, by some counting techniques and character theory. We now deal with the case where $u = 0, v \neq 0$ via Kelley and Owen's result~\cite{KO2017trinomials}. However, it seems still difficult to solve the general case when $u \neq 0, v \neq 0$, which is left as an open problem.

Let $f(x)=x^{d_1}$ and $g(x)=x^{d_2}$ be two power permutations over $\mathbb{F}_p$, where $x \in \mathbb{F}_p^{*}$. Then, the formulation of the cross-correlation of $f, g$ at the shift $(u,v)$ is the following:
\begin{equation}
\label{eq_Pp}
\Psi_{f, g}(u, v) = |\{x\in \mathbb{F}_p^{*} | x^{d_1} \bmod p + v = (x+u)^{d_2} \bmod p \}|.
\end{equation}

\begin{theorem}
Let $f(x)=x^{d_1}$ and $g(x)=x^{d_2}$ be two power permutations over $\mathbb{F}_p$, where $p$ is a prime and $\operatorname{gcd}\left(d_1, p-1\right)=\operatorname{gcd}\left(d_2, p-1\right)=1$. Then for $u = 0, v \neq 0$, we have

$$
\max _{v \neq 0} \max \limits _{f,g \in \mathcal{P}_p} \Psi_{f, g}(0, v) \leq \frac{1}{2} + \sqrt{p-1}.
$$

\end{theorem}

\begin{proof}
Suppose that $u=0$ and $v \neq 0$. Eq.~(\ref{eq_Pp}) is of the form
\begin{equation}
\label{equationPp}
x^{d_1} \bmod p  + v = x^{d_2} \bmod p .
\end{equation}

Taking both sides of Eq.~(\ref{equationPp}) modulo $p$, we obtain
\begin{equation}
\label{equationPp_equiv}
x^{d_1}  + v \equiv x^{d_2} \bmod p .
\end{equation}

Using the result of Lemma~\ref{lem-kelley}, we obtain that the number of solutions of Eq.~(\ref{equationPp_equiv}) is less than $\frac{1}{2}+\sqrt{p-1}$. Note that we can use Lemma~\ref{lem-kelley} because $\operatorname{gcd}(d_1, d_2, p-1)=1$.
\end{proof}

With respect to the family $\mathcal{PW}_p$, we give the first nontrivial upper bound of the maximal cross-correlation for arbitrary $u$ and $v=0$. Given that there are some upper bounds of the maximal cross-correlation of $\mathcal{W}_p$ and $\mathcal{P}_p$, it is crucial to get upper bound between an exponential Welch Costas array and a power permutation. Let $f(x)=x^d$, be a power permutation over
$\mathbb{F}_p$, and let $g=W_1^{\exp}(p, \alpha, 0)$ be an exponential Welch Costas array. Then, the formulation of the cross-correlation between a power permutation and an exponential Welch Costas array at $(u,v)$ is~\cite{Ardalani2023Contributions}:
\begin{equation}
\label{eq_PWp}
\Psi_{f, g}(u, v) = |\{x\in \mathbb{F}_p^{*} | x^d \bmod p+v=\alpha^{x-1+u} \bmod p \}|.
\end{equation}

\begin{lemma}
\label{lem_PWp}
Let $f(x)=x^{d}$ be a power permutation and $g(x)=\alpha^{x-1}$ be an exponential Welch Costas array over $\mathbb{F}_p$, where $p$ is a prime, $\alpha$ is a primitive element and $\operatorname{gcd}\left(d,p-1\right)=1$. For arbitrary $u$ and $v=0$, we have

$$
\max _{u} \Psi _{f, g}(u, 0) < (p-1)^{1 / 2} + (p-1)^{1 / 4} + \frac{1}{2}.
$$
\end{lemma}

\begin{proof}
For arbitrary $u$ and $v=0$, Eq.~(\ref{eq_PWp}) can be expressed as:
$$
\Psi_{f, g}(u, 0) = |\{x\in \mathbb{F}_p^{*} | x^d \bmod p=\alpha^{x-1+u} \bmod p \}|.
$$

To proceed, let $e$ denote the modular inverse of $d$ such that $d e \equiv 1(\bmod \enspace p-1)$. By raising both sides of the congruence $x^d \bmod p=\alpha^{x-1+u} \bmod p$ to the power of $e$, we obtain

$$
(x^d)^e \bmod p=(\alpha^{x-1+u})^e\bmod p.
$$

Simplifying the above yields,
$$
x \bmod p = (\alpha^e)^{x-1+u} \bmod p.
$$

Note that $\alpha^e$ is still a primitive element over $\mathbb{F}_p$ and define $\beta=\alpha^e$. Substituting $\beta$ into the equation, we rewrite it as
$$
x = \beta^{x-1+u} \bmod p.
$$

Observe that the function $f(x)=\beta^{x-1+u} \text{ where } x \in [p-1]$ generates an exponential Costas array $W_1^{\exp}(p, \beta, u)$. The equation $x = \beta^{x-1+u} \bmod p$ implies that $x$ is a fixed point on the main diagonal of the corresponding Costas array. Consequently,  the number of solutions is equal to the number of fixed points on the main diagonal. Since the fixed points on the main diagonal form a Golomb ruler, via Lemma~\ref{size_sidon}, the number of the fixed points is strictly less than $$(p-1)^{1 / 2} + (p-1)^{1 / 4} + \frac{1}{2}.$$
This completes the proof.
\end{proof}

\begin{theorem}
\label{thm_PWp}
Let $f,g$ be two elements of the union family of power permutations and exponential Welch Costas arrays, and let $t$ be the smallest prime divisor of $\frac{p-1}{2}$. For prime $p
\geq 5$, arbitrary $u$ and $v = 0$, we have
$$
\max _{u}\max_{\substack{f, g \in \mathcal{PW}_p }} \Psi_{f, g}(u, 0)\begin{cases}\leq (p-1)^{1 / 2}+(p-1)^{1 / 4}+\frac{1}{2} & \text { if } p \text { is a safe prime, } \\ =\frac{p-1}{t} & \text { otherwise. }\end{cases}
$$
\end{theorem}

\begin{proof}
We consider the following three cases for the selection of $f$ and $g$.

Case (i): Suppose that $f$ and $g$ are both exponential Costas arrays. According to~\cite[Theorem 1]{DG2011Maximal}, we have
$$
\max \limits_{u}\max \limits_{\substack{f, g \in \mathcal{W}_p}} \Psi_{f, g}(u, 0) \leq 2.
$$
if $p$ is a safe prime, and
$$
\max \limits_{u}\max \limits_{\substack{f, g \in \mathcal{W}_p}} \Psi_{f, g}(u, 0) =  \frac{p-1}{t}.
$$
otherwise.

Case (ii): Suppose that $f$ and $g$ are both power permutations. From~\cite[Theorem 5.7 and Theorem 5.9]{Ardalani2023Contributions}, it follows that, for $(u,v)=(0,0)$,
$$
\max \limits_{\substack{f, g \in \mathcal{P}_p \\ f \neq g}} \Psi_{f, g}(0, 0) =  \frac{p-1}{t},
$$
and, for $u \neq 0$ and $v=0$,
$$
\max \limits_{u \neq 0}\max \limits_{\substack{f, g \in \mathcal{P}_p \\ f \neq g}} \Psi_{f, g}(u, 0) \le \left\lceil\frac{p-2}{p-1}(1+\sqrt{p})\right\rceil.
$$
Then we have
$$
\max \limits_{u}\max \limits_{\substack{f, g \in \mathcal{P}_p \\ f \neq g}} \Psi_{f, g}(u, 0) \leq  \frac{p-2}{p-1}(1+\sqrt{p}),
$$
if $p$ is a safe prime, and 
$$
\max \limits_{u}\max \limits_{\substack{f, g \in \mathcal{P}_p \\ f \neq g}} \Psi_{f, g}(u, 0) =  \frac{p-1}{t},
$$
otherwise, since $$\frac{p-1}{t} \geq \sqrt{2(p-1)} \geq \frac{p-2}{p-1}(1+\sqrt{p}) \text { for } p \geq 11.$$

Case (iii): Suppose that $f$ is a power permutation and $g$ is an exponential Costas array. From Lemma~\ref{lem_PWp}, we have
$$
\max \limits_{u}\max \limits_{\substack{f, g \in \mathcal{PW}_p }} \Psi_{f, g}(u, 0) \le (p-1)^{1 / 2} + (p-1)^{1 / 4} + \frac{1}{2},
$$
if $p$ is a safe prime, and 
$$
\max \limits_{u}\max \limits_{\substack{f, g \in \mathcal{PW}_p }} \Psi_{f, g}(u, 0) \le \frac{p-1}{t},
$$
otherwise, since
$$
\frac{p-1}{t} \geq \sqrt{2(p-1)} \geq (p-1)^{1 / 2} + (p-1)^{1 / 4} + \frac{1}{2} \text { for } p \geq 67,
$$
including the numerical results in Table~\ref{PW}.

Therefore, the upper bound on the maximal cross-correlation of the family $\mathcal{PW}_p$ is obtained for arbitrary $u$ and $v = 0$.
\end{proof}

Akin to the consideration of the family of logarithmic Welch Costas arrays, we study the new family combining the logarithmic Welch Costas arrays and power permutations together.
\begin{construction}
For a prime $p \geq 5$, the family $\mathcal{PW}^l_p$ of order $p-1$, is defined as
$$
\mathcal{PW}^l_p=\mathcal{W}^l_p \bigcup \mathcal{P}_p,
$$
where $\mathcal{W}^l_p$ and $\mathcal{P}_p$ are defined in Construction~\ref{const-wpl} and Construction~\ref{const-pp}, respectively. The size of the family $\mathcal{PW}^l_p$ equals $2\phi(p-1)-1$, where $\phi$ is Euler's totient function.
\end{construction}

In fact, the union family of power permutations and logarithmic Welch Costas arrays $\mathcal{PW}^l_p$ is indeed the transpose of the union family of power permutations and exponential Welch Costas arrays $\mathcal{PW}_p$. We easily have the following theorem, via Theorem~\ref{thm_PWp}.

\begin{theorem}
Let $f,g$ be two elements of the union family of power permutations and logarithmic Welch Costas arrays, and let $t$ be the smallest prime divisor of $\frac{p-1}{2}$. For prime $p
\geq 5$, $u = 0$ and arbitrary $v$, we have
$$
\max _{v}\max_{\substack{f, g \in \mathcal{PW}^{l}_p }} \Psi_{f, g}(0, v)\begin{cases}\leq (p-1)^{1 / 2}+(p-1)^{1 / 4}+\frac{1}{2} & \text { if } p \text { is a safe prime, } \\ =\frac{p-1}{t} & \text { otherwise. }\end{cases}
$$
\end{theorem}

\section{Conclusion and future work}\label{sec-con}

In this paper, we investigated several large-size families of Costas arrays and extended arrays, and gave certain bounds on the maximal cross-correlation of these families. More precisely, we gave an upper bound on the maximal cross-correlation of the family $\mathcal{W}^{el}_p$ for arbitrary $u$ and $v$. Additionally, we provided an upper bound for the power permutation family $\mathcal{P}_p$ in the case $u = 0, v \neq 0$, utilizing a result on trinomials over finite fields~\cite{KO2017trinomials}. Furthermore, through combinatorial arguments, we obtained the first nontrivial upper bound on the maximal cross-correlation of the family $\mathcal{P W}_p$ for arbitrary $u$ and $v = 0$, thereby answering an open problem posed by Ardalani~\cite[p.~117]{Ardalani2023Contributions}. A possible future research direction can be devoted to tighten the upper bounds on the existing families of Costas arrays, and present new large-size families.

\bibliographystyle{IEEEtran}
\bibliography{fca-rev}

@article {GW2020Note,
    AUTHOR = {G\'omez-P\'erez, Domingo and Winterhof, Arne},
    TITLE = {A note on the cross-correlation of {C}ostas permutations},
    JOURNAL = {IEEE Trans. Inf. Theory},
    FJOURNAL = {Institute of Electrical and Electronics Engineers.
              Transactions on Information Theory},
    VOLUME = {66},
    YEAR = {2020},
    NUMBER = {12},
    PAGES = {7724--7727},
}

@article {DG2011Maximal,
    AUTHOR = {Drakakis, Konstantinos and Gow, Roderick and Rickard, Scott
              and Sheekey, John and Taylor, Ken},
     TITLE = {On the maximal cross-correlation of algebraically constructed
              {C}ostas arrays},
    JOURNAL = {IEEE Trans. Inf. Theory},
    FJOURNAL = {Institute of Electrical and Electronics Engineers.
              Transactions on Information Theory},
    VOLUME = {57},
    YEAR = {2011},
    NUMBER = {7},
    PAGES = {4612--4621},
}

@article{FL1985Two,
   AUTHOR = {Freedman, A. and Levanon, N.},
   TITLE = {Any two ${N} \times {N}$ {C}ostas signals must have at least one common ambiguity sidelobe if {N} $>$ 3—{A} proof}, 
   JOURNAL = {Proc. IEEE},
   FJOURNAL = {Proceedings of the IEEE},
   VOLUME = {73},
   YEAR = {1985},
   NUMBER = {10},
   PAGES = {1530-1531},
}

@article {Golomb1984Algebraic,
    AUTHOR = {Golomb, Solomon W.},
     TITLE = {Algebraic constructions for {C}ostas arrays},
   JOURNAL = {J. Combin. Theory Series A},
  FJOURNAL = {Journal of Combinatorial Theory. Series A},
    VOLUME = {37},
      YEAR = {1984},
    NUMBER = {1},
     PAGES = {13--21},
}

@article{TR2011Matlab,
  TITLE = {Costas {A}rrays: {S}urvey, {S}tandardization, and {MATLAB} {T}oolbox},
  AUTHOR = {Taylor, Ken and Rickard, Scott and Drakakis, Konstantinos},
  JOURNAL = {ACM Trans. Math. Softw.},
  VOLUME = {37},
  NUMBER = {4},
  PAGES = {1-31},
  YEAR = {2011}
}

@article {GS2010Cycles,
    AUTHOR = {Glebsky, Lev and Shparlinski, Igor E.},
     TITLE = {Short cycles in repeated exponentiation modulo a prime},
   JOURNAL = {Des. Codes Cryptogr.},
  FJOURNAL = {Designs, Codes and Cryptography. An International Journal},
    VOLUME = {56},
      YEAR = {2010},
    NUMBER = {1},
     PAGES = {35--42},
}

@article {DG2009Symmetry,
    AUTHOR = {Drakakis, Konstantinos and Gow, Rod and O'Carroll, Liam},
     TITLE = {On the symmetry of {W}elch- and {G}olomb-constructed {C}ostas
              arrays},
   JOURNAL = {Discrete Math.},
  FJOURNAL = {Discrete Mathematics},
    VOLUME = {309},
      YEAR = {2009},
    NUMBER = {8},
     PAGES = {2559--2563},
}

@ARTICLE{GT1984Constructions,
  AUTHOR = {Golomb, S.W. and Taylor, H.},
  JOURNAL = {Proc. IEEE}, 
  FJOURNAL = {Proceedings of the IEEE},
  TITLE = {Constructions and properties of {C}ostas arrays}, 
  YEAR = {1984},
  VOLUME = {72},
  NUMBER = {9},
  PAGES = {1143-1163}
}

@article {LW2024cross-correlation,
    AUTHOR = {Liu, Huaning and Winterhof, Arne},
     TITLE = {On the cross-correlation of {G}olomb {C}ostas permutations},
   JOURNAL = {IEEE Trans. Inf. Theory},
  FJOURNAL = {Institute of Electrical and Electronics Engineers.
              Transactions on Information Theory},
    VOLUME = {70},
      YEAR = {2024},
    NUMBER = {11},
     PAGES = {7848--7852},
}

@article {WCS2023,
    AUTHOR = {Warnke, Lutz and Correll, Bill and Swanson, Christopher N.},
     TITLE = {The density of {C}ostas arrays decays exponentially},
   JOURNAL = {IEEE Trans. Inf. Theory},
  FJOURNAL = {Institute of Electrical and Electronics Engineers.
              Transactions on Information Theory},
    VOLUME = {69},
      YEAR = {2023},
    NUMBER = {1},
     PAGES = {575--581},
}

@article {KO2017trinomials,
    AUTHOR = {Kelley, Zander and Owen, Sean W.},
     TITLE = {Estimating the number of roots of trinomials over finite
              fields},
   JOURNAL = {J. Symbolic Comput.},
  FJOURNAL = {Journal of Symbolic Computation},
    VOLUME = {79},
      YEAR = {2017},
     PAGES = {108--118},
}

@phdthesis{Ardalani2023Contributions,
  TITLE = {Contributions to the theory of {C}ostas arrays},
  AUTHOR = {Ardalani, Ali},
  SCHOOL = {Otto-von-{G}uericke {U}niversity {M}agdeburg, {G}ermany},
  YEAR = {2023}
}

@book{LN1997finitefield,
  author    = {Lidl, Rudolf and Niederreiter, Harald},
  title     = {Finite {F}ields (\rm{{E}ncyclopedia of {M}athematics and {I}ts {A}pplications)}},
  edition   = {2nd},
  address   = {Cambridge, U.K.},
  publisher = {Cambridge Univ. Press},
  year      = {1997}
}

@techreport{Costas1965Sonar,
  author       = {J. P. Costas},
  title        = {Medium Constraints on Sonar Design and Performance},
  institution  = {GE Co.},
  type         = {Technical Report},
  number       = {Class 1 Rep. R65EMH33},
  year         = {1965}
}

@inproceedings{AP2022Transformation,
  author    = {Ardalani, Ali and Pott, Alexander},
  title     = {A new {T}ransformation for {C}ostas arrays},
  booktitle = {Proc. 10th Int. Workshop Signal Des. Appl. Commun.},
  year      = {2022},
  pages     = {1--5}
}

@article {DT1991Cross-correlation,
    AUTHOR = {Drumheller, David Mark and Titlebaum, Edward L.},
     TITLE = {Cross-correlation properties of algebraically constructed
              {C}ostas arrays},
   JOURNAL = {IEEE Trans. Aerosp. Electron. Syst.},
  FJOURNAL = {Institute of Electrical and Electronics Engineers.
              Transactions on Aerospace and Electronic Systems},
    VOLUME = {27},
      YEAR = {1991},
    NUMBER = {1},
     PAGES = {2--10},
}

@article{Drakakis2011Open,
    AUTHOR = {Drakakis, Konstantinos},
    TITLE = {Open problems in {C}ostas arrays},
    URL = {http://arxiv.org/abs/1102.5727},
    YEAR = {2011} 
}

@article {Cilleruelo2010sidon,
    AUTHOR = {Cilleruelo, Javier},
     TITLE = {Sidon sets in {$\Bbb N^d$}},
   JOURNAL = {J. Combin. Theory Series A},
  FJOURNAL = {Journal of Combinatorial Theory. Series A},
    VOLUME = {117},
      YEAR = {2010},
    NUMBER = {7},
     PAGES = {857--871},
}

@article {JW2014Deficiency,
    AUTHOR = {Jedwab, Jonathan and Wodlinger, Jane},
     TITLE = {The deficiency of {C}ostas arrays},
   JOURNAL = {IEEE Trans. Inf. Theory},
  FJOURNAL = {Institute of Electrical and Electronics Engineers.
              Transactions on Information Theory},
    VOLUME = {60},
      YEAR = {2014},
    NUMBER = {12},
     PAGES = {7947--7954},
}

@article {JY2018cubes,
    AUTHOR = {Jedwab, Jonathan and Yen, Lily},
     TITLE = {Costas cubes},
   JOURNAL = {IEEE Trans. Inf. Theory},
  FJOURNAL = {Institute of Electrical and Electronics Engineers.
              Transactions on Information Theory},
    VOLUME = {64},
      YEAR = {2018},
    NUMBER = {4},
     PAGES = {3144--3149},
}

@article{Costas1984study,
  TITLE={A study of a class of detection waveforms having nearly ideal range-{D}oppler ambiguity properties},
  AUTHOR={J. P. Costas},
  JOURNAL = {Proc. IEEE},
  FJOURNAL = {Proceedings of the IEEE},
  VOLUME={72},
  NUMBER={8},
  PAGES={996--1009},
  YEAR={1984}
}

@article {GT1982Two,
    AUTHOR = {Golomb, S.W. and Taylor, H.},
    TITLE = {Two-dimensional synchronization patterns for minimum ambiguity},
    JOURNAL = {IEEE Trans. Inf. Theory},
    FJOURNAL = {Institute of Electrical and Electronics Engineers.
              Transactions on Information Theory},
    VOLUME = {28},
    YEAR = {1982},
    NUMBER = {4},
    PAGES = {600--604},
}

\end{document}